\def\masterthesis{0}
\def\cryptology{0} % For journal of cryptology
\def\sigconf{0} %ACM SIG conferfence
\def\big{0} % For using a huge font

\def\draft{0}
\def\anon{0}
\def\shownomenclature{0}
\def\smallbib{1} %This should only be used if not submitted to a journal/conference with proceedings, as these use their own format
\def\toc{0} %Show TOC. Should be set to 1 if paper is long (say, above 30 pages), and for all types of thesis.

\ifnum\big=1
    \documentclass[final,17pt]{extarticle} \usepackage{fullpahttps://www.overleaf.com/project/631831f15d6949f227c91ddcge}
\else
    \ifnum\sigconf=1
        \documentclass[sigconf,final]{acmart}
    \else
        \ifnum\draft=1
           \documentclass[11pt,draft,a4paper]{article}
        \else
            \ifnum\masterthesis=1
                \documentclass[11pt,final,a4paper,titlepage]{article}
            \else
                \documentclass[11pt,final,a4paper]{article}
            \fi
        \fi
        \usepackage{lmodern}
        \usepackage[in]{fullpage}
    \fi
\fi

\input{macros}

\newcommand{\agnote}[1]{\authnote{Alex}{#1}{brown}}

\usepackage{amsmath}
\usepackage{mathtools}
\usepackage{algorithm,algorithmicx,algpseudocode}
\usepackage[normalem]{ulem}
\usepackage{float}

\usepackage{xspace}
\DeclareMathAlphabet{\mathpzc}{OT1}{pzc}{m}{it}

\newcommand{\dk}{\ensuremath{\mathsf{dk}}\xspace}

\newcommand{\gen}{\ensuremath{\mathsf{Gen}}\xspace}

\newcommand{\qpk}{\ensuremath{\mathpzc{qpk}}\xspace} 
\newcommand{\qpkgen}{\ensuremath{\mathpzc{QPKGen}}\xspace}
\newcommand{\qenc}{\ensuremath{\mathpzc{Enc}}\xspace}
\newcommand{\qdec}{\ensuremath{\mathpzc{Dec}}\xspace}
\newcommand{\del}{\ensuremath{\mathpzc{Del}}\xspace}
\newcommand{\qc}{\ensuremath{\mathpzc{qc}} \xspace}

 \newcommand{\prfspd}{\text{Pseudorandom function-like state generator with proofs of destruction}} %%
 \newcommand{\PRSPD}{\textsf{PRSPD}}%textabbrevbol was throwing errors

 \newcommand{\PRFSPD}{\textsf{PRFSPD}}
 
 \newcommand{\ver}{\ensuremath{\mathpzc{Ver}}\xspace}

 \newcommand{\Test}{\ensuremath{\mathsf{Test}}}

 \newcommand{\secproof}{\ensuremath{\mathsf{Unclonability\text{-}of\text{-}proofs}}\xspace}
 \newcommand{\experiment}[2]{\ensuremath{\mathsf{Exp}^{#1,#2}_{\secpar}}\xspace}
 \newcommand{\clon}[2]{\ensuremath{\mathsf{Clonning}\text{-}\experiment{#1}{#2}}\xspace}

 \newcommand{\haar}{\ensuremath{\mathpzc{Haar}}\xspace}
 \newcommand{\uniform}[1]{\ensuremath{\xleftarrow{u}\{0,1\}^{#1}}\xspace}

 \newcommand{\qpt}{\textsf{QPT}}

\begin{document}
\ifnum\masterthesis=0
    \title{Encryption with Quantum Public Keys
    %\ifdraft{\\(working draft)}
    }
\fi

\ifnum\anon=0
    \ifnum\masterthesis=0
        \ifnum\sigconf=0
            \ifnum\cryptology=1
                \author{Or Sattath}
                \affil{Department of Computer Science, Ben Gurion University of the Negev, Beersheba, Israel\\
                sattath@bgu.ac.il}
            \else
                \author[1]{Alex B. Grilo}
                \affil[1]{Sorbonne Universit\'e, CNRS, LIP6}
                \author[2]{Or Sattath}
                \affil[2]{Computer Science Department, Ben-Gurion University of the Negev}
                \author[1]{Quoc-Huy Vu}
                \date{}
            \fi
        \else
            \author{Or Sattath}
            \affiliation{%
            \institution{Computer Science Department, Ben-Gurion University of the Negev}
            \country{Israel}}
        \fi
    \fi
\else
    \ifnum\sigconf=0
        \author{}
    \fi
\fi

\ifnum\masterthesis=0
    \ifnum\sigconf=0
        \maketitle
    \fi
\fi

\ifnum\masterthesis=1
    \begin{titlepage}
        \centering
        { Ben-Gurion University of the Negev}
        
        {The Faculty of Natural Sciences}
        
        {\small The Department of Computer Science}
        
        \vspace{2cm}
        
        \includegraphics[scale=\logoscale]{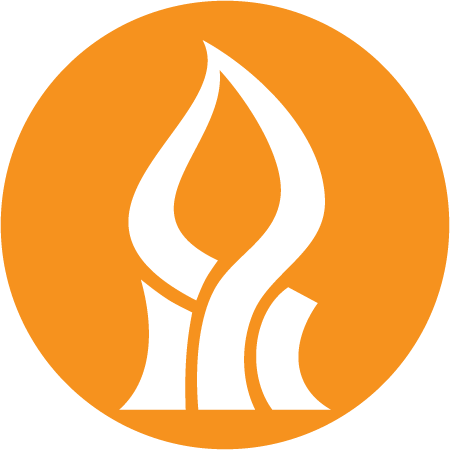}
        
        \vspace{2cm}
        
        {\Large \bfseries Template of Thesis}
        
        \vspace{2cm}
        
        {\small Thesis submitted in partial fulfillment of the requirements for the Master of Sciences degree}
        
        \vspace{1cm}
        
        {\bfseries Name of Student}
        
        {Under the supervision of Dr. Or Sattath}
        
        \vspace{2cm}
        
        \today
    \end{titlepage}
    
    \begin{titlepage}
        \centering
        { Ben-Gurion University of the Negev}
        
        {The Faculty of Natural Sciences}
        
        {\small The Department of Computer Science}
        
        \vspace{2cm}
        
        {\Large \bfseries Template of Thesis}
        
        \vspace{2cm}
        
        {\small Thesis submitted in partial fulfillment of the requirements for the Master of Sciences degree}
        
        \vspace{1cm}
        
        {\bfseries Name of Student}
        
        {Under the supervision of Dr. Or Sattath}
        
        \vspace{1cm}
        
        {\small Signature of student: \longunderline Date: \longunderline}
        
        \vspace{0.5cm}
        
        {\small Signature of supervisor: \longunderline Date: \longunderline}
        
        \vspace{0.5cm}
        
        \begin{changemargin}{-1.5cm}{-1.5cm}
        \centering
            {\small Signature of chairperson of the committee for graduate studies: \longunderline Date: \longunderline}
        \end{changemargin}
        \vspace{2cm}
        
        \today
    \end{titlepage}

    \pagenumbering{roman}
    \begin{center}
        {\large \bfseries Title of Thesis}
        
        \vspace{0.5cm}
        
        {\bfseries Name of Student}
        
        \vspace{0.5cm}
        
        Thesis submitted in partial fulfillment of the requirements for the Master of Sciences degree
        
        \vspace{0.25cm}
        
        {Ben-Gurion University of the Negev}
        
        \vspace{0.25cm}
        
        \today
        
        \vspace{2cm}
        
        {\bfseries Abstract}
        
        \vspace{0.5cm}
    \end{center}
\else 
    \ifnum\sigconf=0
        \begin{abstract}
    \fi
\fi
\ifnum\sigconf=0
%IN ARXIV SUBMISSIONS, REMEMBER TO INDENT NEW PARAGRAPHS SINCE NEWLINES ARE IGNORED. 
It is an important question to find constructions of quantum cryptographic protocols which rely on weaker computational assumptions than classical protocols. Recently, it has been shown that oblivious transfer and multi-party computation can be constructed from one-way functions, whereas this is impossible in the classical setting in a black-box way. In this work, we study the question of building quantum public-key encryption schemes from one-way functions and even weaker assumptions. Firstly, we revisit the definition of IND-CPA security to this setting. Then, we propose three schemes for quantum public-key encryption from one-way functions, pseudorandom function-like states with proof of deletion and pseudorandom function-like states, respectively.

\fi 
\ifnum\masterthesis=0
    \ifnum\sigconf=0
        \end{abstract}
    \fi
\fi

\ifnum\sigconf=1
    \begin{abstract}
        Abstract of SigConf goes here.
    \end{abstract}
    \keywords{}
    \maketitle
\fi

\ifnum\masterthesis=1
    \pagebreak
    \pagenumbering{arabic}
    \subsection*{Acknowledgments}
    \setcounter{tocdepth}{3}
    
    \pagebreak
    \tableofcontents
    \listoffigures % Remove if there are no figures
    \listoftables % Remove if there are no tables
    \pagebreak
\fi

\ifnum\toc=1
    \tableofcontents 
\fi

\section{Introduction} % (fold)
\label{sec:introduction}

The use of quantum resources to enable cryptographic tasks under weaker
assumptions (or even {\em unconditionally}) than classically were actually the
first concrete proposals of quantum computing, with the seminal quantum money
protocol of Wiesner~\cite{Wie83} and the key-exchange protocol of Bennet and
Brassard~\cite{BB84}.

Since then, it became a fundamental question in the field of quantum
cryptography to find other primitives that can be implemented under weaker
computational assumptions. It has recently shown that there exist quantum
protocols for Oblivious Transfer (and therefore arbitrary multi-party
computation (MPC)) based on the existence of one-way functions
(OWF)~\cite{CKM21a,GLSV21}.  Moreover, the proposed protocols use simple quantum
technology that is available currently.  The assumption to construct these
primitives has been recently improved by showing that the existence of
pseudo-random states (PRS) is sufficient for such primitives.  Notice that
existence of PRS is plausibly a strictly weaker assumption than the existence of
OWF, given that PRS families can be constructed from OWF in a black-box way~\cite{JLS18}, and
we have oracle separations between PRS and OWF~\cite{Kre21,KQST22}.

We notice that classically, OT and MPC are ``Cryptomania'' objects, meaning that they can be
constructed from assumptions that imply public-key encryption (PKE), but there are oracle separations
between OWF and  PKE (and thus OT and MPC)~\cite{IR89}. Therefore, we do not expect that
such powerful objects can be built classically from OWF. In this work, we investigate the following natural question: 
\begin{tcolorbox}
\begin{center}
    Can we have quantum protocols for public-key encryption, the heart of
    Cryptomania, based on post-quantum one-way functions, or even weaker assumptions?
\end{center} 
\end{tcolorbox}

Recent results imply that it is improbable to achieve PKE schemes from OWF if the public-key and ciphertext are classical even if the encryption or decryption algorithms are quantum~\cite{AustrinCCFLM22}. Therefore, in this work, we will consider schemes where the public-key or ciphertext are quantum states.

We notice that the first problem that we need to address is the syntax of
quantum public-key encryption (qPKE) and the corresponding security games. We
need to provide a general definition for qPKE where both the public-key and ciphertext
might be general quantum states. 
Furthermore, we note that if the public-key is a quantum state, it might be measured, and the ciphertexts might depend on the measurement outcome.
This motivates a stronger definition in which the adversary gets oracle access to the encryption, which we call IND-CPA-EO security.

With our new security definition in hand, we propose three protocols for
implementing qPKE from OWF or weaker assumptions,  with different
advantages and disadvantages. More concretely, we show the following:
\begin{enumerate}
    \item Assuming the existence of post-quantum OWFs, there exists a qPKE scheme
      with quantum public-keys and classical ciphertexts that is IND-CPA-EO
      security.
    \item Assuming the existence of pseudo-random function-like states with
      proof of destruction (PRFSPDs), there is a qPKE scheme with quantum public-key and
      classical ciphertext that is IND-CPA-EO secure.
    \item Assuming the existence of
      pseudo-random function-like states (PRFSs) with super-logarithmic
      input-size\footnote{Note that PRS implies PRFS with logarithmic size inputs.
      No such implication is known for super-logarithmic inputs.}, there is a
      qPKE scheme with quantum public-key and quantum ciphertext. In this
      scheme, the quantum public key enables the encryption of a single
      message.
\end{enumerate}

We would like to stress that for the first two constructions, even if the public-key is a quantum state, the
ciphertexts are classical and one quantum public-key can be used to encrypt
multiple messages.
Our third construction shows how to construct quantum public key encryption from assumptions (the existence of pseudorandom function-like states) which are potentially weaker than post-quantum one-way functions, but the achieved protocol only allows the encryption of one message per public-key.

We would also like to remark that it has been recently shown that OWFs imply PRFSs with super-logarithmic input-size~\cite{AQY21} and PRFSPDs~\cite{BBSS23}.
Therefore, the security of the second and third protocols is based on a potentially weaker
cryptographic assumption than the first one.
Furthermore, PRFSs with super-logarithmic input-size is \emph{separated} from one-way functions~\cite{Kre21}; therefore, our third result shows a black-box separation between a certain form of quantum public key encryption and one-way functions.

However, the first protocol is much simpler to describe
and understand since it only uses standard (classical) cryptographic objects. Moreover, it is the only scheme that achieves perfect correctness.

\subsection{Technical overview}
\label{sec:technical_overview}
In this section we provide a technical overview of our results. In  \Cref{sec:intro-definition}, we explain the challenges and choices in order to define qPKE and its security definition. In \Cref{sec:intro-constructions}, we present our constructions for qPKE. 

\subsubsection{Definitions of qPKE and IND-CPA-EO}
\label{sec:intro-definition}

In order to consider public-key encryption schemes with quantum public-keys, we
need to first revisit the security definition and we define a new security game that we call IND-CPA-EO. 

In the
classical-key case (even with quantum ciphertexts), the adversary is given a
copy the public-key $\pk$ and therefore
is able to run the encryption algorithm $\enc(\pk, \cdot)$ as many times they want (where the only
constraint is that the adversary is polynomially bounded), and just then choose
the messages $m_0$ and $m_1$ in the IND-CPA game.

In the quantum public-key case, the first issue is that $\enc(\rho_{\qpk},\cdot)$ might
consume the quantum public-key $\rho_{\qpk}$.  Moreover, having more copies of the quantum state could leak
information to the adversary (which cannot be the case in the classical-key case,
since the adversary {\em can} copy $\pk$). Therefore, the first modification in IND-CPA towards IND-CPA-EO is that the adversary is given multiple
copies of the public-key $\rho_{\qpk}$.

Secondly, it could be the case that $\enc(\rho_\qpk,\cdot)$ measures $\rho_\qpk$ and
modifying it to $\rho_\qpk'$, which is different from $\rho_\qpk$ but still a valid
key that enables encryption. The second modification that we need is syntactic: $\enc(\rho_\qpk, m)$ outputs $(\rho_\qpk',c)$, where $c$ is used as the ciphertext
and $\rho_\qpk'$ is used as the key to encrypt the next message. 

We notice that in the previous discussion, it could also be the case that $\enc$
measures $\rho_{\qpk}$ and the measurement outcome is used to encrypt all the
measurements. In this case, even if the adversary is given multiple copies of
$\rho_{\qpk}$, they would not be able to post-select on the same measurement
outcome and the distribution of ciphertexts that they could generate would be
different from the ones in the game. Therefore, we consider a stronger notion of
security where the adversary has also access to the encryption oracle that will
be used in the IND-CPA-EO game. We notice that this issue does not make
sense classically since the public-key used by the challenger and adversary is
exactly the same and the distribution of the ciphertexts would be the same.

Finally, we  would like to mention a few problems with having $\rho_{\qpk}$ as a mixed state. Firstly, there is no efficient way of comparing if
the two given public-keys are the same, preventing an honest decryptor to ``compare'' a purported public-key, whereas for pure states, this can be achieved using the swap-test. 
Secondly, if mixed states are allowed, then the notions of symmetric and public key encryption coincide, both in the classical and quantum setting: Consider a symmetric encryption scheme $(\mathsf{SKE}.\keygen, \mathsf{SKE}.\enc,\mathsf{SKE}.\dec)$. We can transform it into a public-key scheme. To generate the keys, we use the output of $\mathsf{SKE}.\keygen$ as the secret-key and use it to create the uniform mixture $\frac{1}{2^n} \sum_{\in\{0,1\}^n}\ketbra{x}\tensor \ketbra{Enc_{sk}(x)}$ as public-key. The ciphertext of a message $m$ is $(\enc_x(m),\enc_{sk}(x))$. To decrypt, the decryptor would first recover $x$ by decrypting the second element in the ciphertext using $sk$, and then recover $m$ by decrypting the first item using $x$ as the secret key. 
Therefore,
we have that the meaningful notion of PKE with quantum public-keys should
consider only pure states as quantum public-keys.

\subsubsection{Constructions for qPKE}
\label{sec:intro-constructions}
As previously mentioned, we propose in this work three schemes for qPKE, based on different assumptions. 

\paragraph{QPKE from OWFs.} Our first scheme is based on the existence of post-quantum pseudo-random functions (PRF) and
post-quantum IND-CPA secure symmetric-key encryption schemes, and both of these primitives can be
constructed from post-quantum OWFs. More concretely, let $\{f_k\}_k$ be a keyed PRF family and $(\mathsf{SE}.\enc,\mathsf{SE}.\dec)$ be a symmetric-key encryption scheme. The
secret key $\dk$ of our scheme is a uniformly random key for the PRF, and for a fixed $\dk$,
the quantum public-key state is 
\begin{align}\label{eq:public-key-owf}
\ket{\qpk_\dk}=\frac{1}{\sqrt{2^\secpar}} \sum_{x \in
  \{0,1\}^\secpar}\ket{x}\ket{f_\dk(x)}.
\end{align}
For clarity, we will drop the index of $\ket{\qpk}$ when $\dk$ is clear from the context. 

The encryption algorithm will then measure $\ket{\qpk}$ in the computation
basis leading to the outcome $(x^*, f_{\dk}(x^*))$. The ciphertext of a message
$m$ is $(x^*,\mathsf{SE}.\enc_{f(x^*)}(m))$ and
the decryption algorithm receives as input a ciphertext $(\hat{x},\hat{c})$ and
outputs $\mathsf{SE}.\dec_{f(\hat{x})}(\hat{c})$.

Using a hybrid argument, we prove that any adversary that breaks the
security of this qPKE scheme can be used to break the security of the PRF family
or the security of the symmetric-key encryption scheme.
The formal construction and its proof of security is given in~\cref{sect:qpke-from-owf}.

We notice that such a scheme allows the encryption/decryption of many messages using the same
measurement outcome $(x^*, f_{\dk}(x^*))$.

\paragraph{QPKE from PRFSPDs.} Our second scheme is based on pseudo-random function-like states with proof of destruction (PRFSPDs), which was recently defined in~\cite{BBSS23}. A family of states $\{\ket{\psi_{k,x}}\}_{k,x}$ is pseudo-random function-like~\cite{AQY21} if 
\begin{enumerate}
\item there is a quantum polynomial-time algorithm $\gen$ such that 
\[\gen(k,\sum_{x} \alpha_x \ket{x}) = \sum_{x} \alpha_x \ket{x}\ket{\psi_{k,x}} \text{, and}\] 
  \item no \(\qpt\) adversary can distinguish $(\ket{\psi_1},...,\ket{\psi_\ell})$ from
    $(\ket{\phi_1},...,\ket{\phi_\ell})$, where \allowbreak
    $\ket{\psi_i} = \sum_{x} \alpha^i_x \ket{x}\ket{\psi_{k,x}}$,
    $\ket{\phi_i} = \sum_{x} \alpha^i_x \ket{x}\ket{\phi_{x}}$ and $\{\ket{\phi_x}\}_x$ are Haar
    random states and the states $\ket{\sigma_i} = \sum_{x} \alpha^i_x \ket{x}$ are chosen by
    the adversary.
\end{enumerate}
Recently, \cite{BBSS23} extended this notion to pseudo-random function-like states with proof of destruction, where we have two algorithms $\del$ and $\ver$, which allows us to verify if a copy of the PRFS was deleted.

We will discuss now how to provide the one-shot security\footnote{Meaning that
one can only encrypt once using $\ket{\qpk}$.} of the encryption of a one-bit
message and we discuss later how to use it to achieve general security.

The quantum public-key in this simplified case is
\begin{align}\label{eq:public-key-prfspd}
  \frac{1}{\sqrt{2^\secpar}} \sum_{x \in \{0,1\}^\secpar}\ket{x}\ket{\psi_{\dk,x}}. 
\end{align}

The encryptor will then measure the first register of $\ket{\qpk}$ and the
post-measurement state is $\ket{x^*}\ket{\psi_{\dk,x^*}}$. The encryptor will
then generate a proof of deletion $\pi = \del(\ket{\psi_{\dk,x^*}})$. The
encryption chooses $r \in \{0,1\}^\lambda$ uniformly at random and
compute the ciphertext $c = (x^*,y)$ where $y =  \begin{cases}r, & \text{if } b = 0 \\ \pi,& \text{if
} b = 1\end{cases}$.

The decryptor will receive some value $(\hat{x},\hat{y})$ and
decrypt the message $\hat{b} =  \ver(\dk,\hat{x}, \hat{y})$.

The proof of the security of such a scheme closely follows the proof of our
first scheme.

Notice that repeating such a process in parallel trivially gives a one-shot
security of the encryption of a string $m$ and moreover, such an encryption is
classical. Therefore, in order to achieve IND-CPA-EO secure qPKE scheme, we can actually
encrypt a secret key $\sk$ that is chosen by the encryptor, and send the message
encrypted under $\sk$. We leave the details of such a construction and its proof
of security to~\cref{sect:qpke-from-prfspd}.

\paragraph{QPKE from PRFSs.} Finally, our third scheme uses the public-key
\begin{align}\label{eq:public-key-prfs}
  \frac{1}{\sqrt{2^\secpar}} \sum_{x \in \{0,1\}^\secpar}\ket{x}\ket{\psi_{\dk,x}},
\end{align}
where $\{\ket{\psi_{k,x}}\}_{k,x}$ is a PRFS family and the size of the input $x$ is super-logarithmic on the security parameter.

The encryptor will then measure the first register of $\ket{\qpk}$ and the
post-measurement state is $\ket{x^*}\ket{\psi_{\dk,x^*}}$. The encryptor will
then compute the ciphertext $c = (x^*,\rho)$ where
\begin{align}
\rho =  \begin{cases} I / 2^{n}, & \text{if } b = 0 \\ \ketbra{\psi_{\dk,x^*}},& \text{if } b = 1\end{cases}.
\end{align}
The decryptor, on ciphertext $c = (\hat{x},\hat{\rho})$, sets $\hat{b} \gets \Test(\dk, \hat{x}, \hat{\rho})$, where $\Test$ is a tester algorithm checking whether $\hat{\rho}$ is the output of the PRFS generator with key $\dk$ and input $\hat{x}$, that is whether $\hat{\rho} = \ketbra{\psi_{\dk,\hat{x}}}$.

The proof of the security of such a scheme also closely follows the one of our first scheme, and we give the formal construction in~\cref{sect:qpke-from-prfs}.

\subsection{Related works} % (fold)
\label{sub:related_works}
Gottesman~\cite{GottesmanConstruction} \agnote{make sure that this is correct} has a candidate construction (without formal security analysis) encryption scheme with quantum public keys and quantum ciphers, which consumes the public key for encryption.
Ref.~\cite{OTU00} defines and constructs a public-key encryption where the keys, plaintexts and ciphers are classical, but the algorithms are quantum. (In their construction, only the key-generation uses Shor's algorithm.) 

In ~\cite{NI09}, the authors define and provide impossibility results regarding encryption with quantum public keys. Classically, it is easy to show that a (public) encryption scheme cannot have deterministic ciphers; in other words, encryption must use randomness. They show that this is also true for a quantum encryption scheme with quantum public keys.

In Ref.~\cite{MY22a,MY22b}, the authors study digital signatures with quantum signatures, and more importantly in the context of this work, quantum public keys. 

\subsection{Concurrent and independent work}

Very recently, two concurrent and independent works have achieved similar tasks. Coladangelo \cite{Col23} shows a qPKE scheme whose construction is very different from ours, and uses a quantum trapdoor function, which is a new notion first introduced in their work. The hardness assumption is the existence of post-quantum OWF. Each quantum public key can be used to encrypt a single message (compared to our construction from OWF, where the public key can be used to encrypt multiple messages). The ciphertexts are quantum (whereas our construction from OWF has classical ciphertexts).  Barooti, Malavolta and Walter~\cite{BMW23} also construct a qPKE scheme based on OWF.  Their construction has quantum public keys and classical ciphertexts, and is very similar to the construction we propose in \cref{sect:qpke-from-owf}. They do not discuss the notion of IND-CPA-EO security, but we believe that the hybrid encryption approach we use to achieve IND-CPA-EO can be used in their construction as well.  Moreover, their construction also achieves CCA security, which is stronger than CPA security.
They leave open the question of constructing qPKE from weaker assumptions than OWF, which we answer affirmatively in our work.

\section{Definitions and Preliminaries}
\label{sec:definitions}
\subsection{Notation}
Throughout this paper, \(\secpar\) denotes the security parameter.
The notation \(\negl\) denotes any function \(f\) such that \(f(\lambda) = \lambda^{-\omega(1)}\), and
\(\poly\) denotes any function \(f\) such that \(f(\lambda) = \mathcal{O}(\lambda^c)\) for some
\(c > 0\).
When sampling uniformly at random a value \(a\) from a set~\(\mathcal{U}\), we
employ the notation \(a \sample \mathcal{U}\).
When sampling a value \(a\) from a probabilistic algorithm \(\adv\), we employ the
notation \(a \leftarrow \adv\).
Let \(\size{\cdot}\) denote either the length of a string, or the cardinal of a finite
set, or the absolute value.
By \(\ppt\) we mean a polynomial-time non-uniform family of probabilistic
circuits, and by \(\qpt\) we mean a polynomial-time family of quantum circuits.

\subsection{Pseudorandom Function-Like State (PRFS) Generators}

The notion of pseudorandom states were first introduced by Ananth, Qian and Yuen
in~\cite{AQY21}.
A stronger definition where the adversary is allowed to make superposition queries to the challenge oracles was introduced in the follow-up work~\cite{AGQY22}. 
We reproduce their definition here:

\begin{definition}[Quantum-accessible PRFS generator]
  \label{def:prfs}
We say that a $\qpt$ algorithm $G$ is a quantum-accessible secure pseudorandom function-like state generator if for all $\qpt$ (non-uniform) distinguishers $A$ if there exists a negligible function $\epsilon$, such that for all $\lambda$, the following holds:
 \[
    \left|\Pr_{k \leftarrow \{0,1\}^{\secparam} }\left[A_\lambda^{\ket{{\cal O}_{\sf PRFS}(k,\cdot)}}(\rho_\lambda) = 1\right] - \Pr_{{\cal O}_{\sf Haar}}\left[A_\lambda^{\ket{{\cal O}_{\sf Haar}(\cdot)}}(\rho_\lambda) = 1\right]\right| \le \epsilon(\lambda),
  \]
  where:
\begin{itemize}
    \item ${\cal O}_{\sf PRFS}(k,\cdot)$, %modeled as a channel, 
    on input a $d$-qubit register ${\bf X}$, does the following: it  applies an isometry channel that controlled on the register ${\bf X}$ containing $x$, it creates and stores  $G_{\secparam}(k,x)$ in a new register ${\bf Y}$.  It outputs the state on the registers ${\bf X}$ and ${\bf Y}$. 
    \item ${\cal O}_{\sf Haar}(\cdot)$, modeled as a  channel, on input a $d$-qubit register ${\bf X}$, does the following: it applies a channel that controlled on the register ${\bf X}$ containing $x$, stores $\ketbra{\vartheta_x}$ in a new register ${\bf Y}$, where $\ket{\vartheta_x}$ is sampled from the Haar distribution. It outputs the state on the registers ${\bf X}$ and ${\bf Y}$. 
\end{itemize}
Moreover, $A_{\secparam}$ has superposition access to ${\cal O}_{\sf PRFS}(k,\cdot)$ and ${\cal O}_{\sf Haar}(\cdot)$ (denoted using the ket notation).
\par We say that $G$ is a $(d(\lambda),n(\lambda))$-QAPRFS generator to succinctly indicate that its input length is $d(\lambda)$ and its output length is $n(\lambda)$.
\end{definition}
Given a state \(\rho\), it is useful to know whether it is the output of a PRFS
generator with key \(k\) and input \(x\).
The following lemma shows the existence of a tester algorithm to test any PRFS
states in a semi-black-box way.

\begin{lemma}{{\cite[Lemma 3.10]{AQY21}}}
  Let $G$ be a $(d, n)$-PRFS generator.
  There exists a \(\qpt\) algorithm $\Test(k, x, \cdot)$, called the tester
  algorithm for $G(k, x)$, such that here exists a negligible function $\nu(\cdot)$
  such that for all $\lambda$, for all $x \neq y$,
  \[
    \Pr_{k}[\Test(k, x, G(k, x)) = 1] \ge 1 - \nu(\lambda),
  \]
  and
  \[
    \Pr_{k}[\Test(k, x, G(k, y)) = 1] \le 2^{-n(\lambda)} + \nu(\lambda).
  \]
\end{lemma}

\subsection{Quantum Pseudorandomness with Proofs of Destruction}
The rest of this section is taken verbatim from~\cite{BBSS23}.
\begin{game}
\caption{$\clon{\adv}{\PRSPD}$ %$\exCerDel_{\adv,\Pi}(\secpar)$:
}
\begin{algorithmic}[1]
    \State Given input $1^\secpar$, Challenger samples $k\gets \{0,1\}^{w(\secpar)}$ uniformly at random.
    \State $\adv$ sends $m$ to the challenger.
    \State Challenger runs $\gen(k)^{\tensor m}$ and sends $\ket{\psi_k}^{\tensor m}$ to $\adv$.
    \State $\adv$ gets classical oracle access to $\ver(k,\cdot)$.
    \State $\adv$ outputs $c_1,c_2,\ldots,c_{m+1}$ to the challenger.
    \State Challenger rejects if $c_i$'s are not distinct.
    \For{$i\in[m+1]$}
    Challenger runs $b_i\gets \ver(k,c_i)$
    \EndFor
    \State Return $\wedge_{i=1}^{m+1} b_i$.
\end{algorithmic}
\label{game:cloning-prspd}
\end{game}

\begin{definition}[\prfspd]
  A $\PRFSPD$ scheme with key-length $w(\secpar)$, input-length $d(\secpar)$, output length $n(\secpar)$ and proof length $c(\secpar)$ is a tuple of  $\qpt$ algorithms $\gen,\del,\ver$ with the following syntax:
  \begin{enumerate}
       \item $\ket{\psi^x_k}\gets \gen(k,x)$: takes a key $k\in \{0,1\}^w$, an input string $x\in \{0,1\}^{d(\secpar)}$, and outputs an $n$-qubit pure state 
        $\ket{\psi^x_k}$.
      \item $p\gets \del(\ket{\phi})$: takes an $n$-qubit quantum state $\ket{\phi}$ as input, and outputs a $c$-bit classical string, $p$.
      \item $b\gets \ver(k,x,q)$: takes a key $k\in \{ 0,1 \}^w$, a $d$-bit input string $x$, a $c$-bit classical string $p$ and outputs a boolean output $b$.
  \end{enumerate}

  \paragraph*{Correctness.} A $\PRFSPD$ scheme is said to be correct if for every $x\in \{0,1\}^{d}$,
  \[\Pr_{k\uniform{w}}[1\gets\ver(k,x,p)\mid p\gets \del(\ket{\psi^x_k}); \ket{\psi^x_k}\gets \gen(k,x)]=1\]
  \paragraph*{Security.}
  \begin{enumerate}
      \item \textsf{Pseudorandomness:} A $\PRFSPD$ scheme is said to be (adaptively) pseudorandom if for any $\qpt$ adversary $\adv$, and any polynomial $m(\secpar)$, there exists a negligible function $\negl$, such that
      \[\left|\Pr_{ k\gets\{0,1\}^w}[\adv^{\gen(k,\cdot)}(1^\secpar)=1]-\Pr_{ \forall x\in \{0,1\}^d, \ket{\phi^x}\gets\mu_{(\CC^2)^{\tensor n}}}[\adv^{\haar^{\{\ket{\phi^x}\}_{x\in \{0,1\}^d}}(\cdot)}(1^\secpar)=1]\right|=\negl,\]
      where $\forall x\in \{0,1\}^d$, $\haar^{\{\ket{\phi^x}\}_{x\in \{0,1\}^d}}(x)$ outputs $\ket{\phi^x}$. Here $\adv^{\gen(k,\cdot)}$ represents that $\adv$ gets classical oracle access to $\gen(k,\cdot)$.\label{item:pseudorandomness}
      \item $\secproof$: A $\PRFSPD$ scheme satisfies $\secproof$ if for any $\qpt$ adversary $\adv$ in cloning game (see Game~\ref{game:cloning-prfspd}), there exists a negligible function $\negl$ such that
      \[\Pr[\clon{\adv}{\PRFSPD}=1]=\negl.\]
      \label{item:security of proof of independence}

\begin{game}
\caption{$\clon{\adv}{\PRFSPD}$ %$\exCerDel_{\adv,\Pi}(\secpar)$:
}
\begin{algorithmic}[1]
    \State Given input $1^\secpar$, Challenger samples $k\gets \{0,1\}^{w(\secpar)}$ uniformly at random.
    \State Initialize an empty set of variables, $S$.
    \State $\adv$ gets oracle access to $\gen(k,\cdot)$,  $\ver(k,\cdot, \cdot)$ as oracle.
    \For{$\gen$ query $x$ made by $\adv$}
    \If{$\exists$ variable $t_x\in S$}
    $t_x=t_x+1$.
    \Else{ Create a variable $t_x$ in $S$, initialized to $1$.}
    \EndIf
    \EndFor
    \State $\adv$ outputs $x,c_1,c_2,\ldots,c_{t_x+1}$ to the challenger.
    \State Challenger rejects if $c_i$'s are not distinct.
    \For{$i\in[m+1]$}
    $b_i\gets \ver(k,x,c_i)$
    \EndFor
    \State Return $\wedge_{i=1}^{m+1} b_i$.
\end{algorithmic}
\label{game:cloning-prfspd}
\end{game}

\end{enumerate}

\label{definition:PRFSPD}
\end{definition}

\section{Security definitions for qPKE}
In this section, we introduce the new notion of encryption with quantum public keys (\cref{def:eqpk}) and present our indistinguishability under chosen-plaintext attacks security for quantum public-key encryption.

\begin{definition}[Encryption with quantum public keys]
\label{def:eqpk}
Encryption with quantum public keys (qPKE) consists of 4 algorithms with the following syntax:
\begin{enumerate}
    \item $\dk \gets \gen(\secparam)$: a $\ppt$ algorithm, which receives the security parameter and outputs a classical decryption key.
    \item $\ket{\qpk}\gets \qpkgen(\dk)$: a $\qpt$ algorithm, which receives a classical decryption key $\dk$, and outputs a quantum public key $\ket{\qpk}$.
    We require that the output is a pure state, and that $t$ calls to $\qpkgen(\dk)$ should yield the same state, that is, $\ket{\qpk}^{\tensor t} $.
    \item $(\qpk',\qc) \gets \qenc(\qpk,m)$: a $\qpt$ algorithm, which receives a quantum public key $\qpk$ and a plaintext $m$, and outputs a (possibly classical) ciphertext $\qc$ and a recycled public-key $\qpk'$.
    \item $m \gets \qdec(\dk, \qc)$: a $\qpt$ algorithm, which uses a decryption key $\dk$ and a ciphertext $\qc$, and outputs a classical plaintext $m$.
    In the case $\qc$ is classical, we consider $\qdec$ as a $\ppt$ algorithm.
\end{enumerate}
\end{definition}
We say that a qPKE scheme is complete if for every message $m \in  \{0,1\}^*$ and any security parameter $\secpar \in \NN$, the following holds:
\[
\Pr\left[
\qdec(\dk, \qc) = m
\ \middle\vert
\begin{array}{r}
\dk\gets \gen(\secparam) \\
\ket{\qpk} \gets \qpkgen(\dk) \\
(\qpk', \qc) \gets \qenc(\ket{\qpk}, m)
\end{array}
\right] \geq 1 -\negl,
\]
where the probability is taken over the randomness of $\gen$, $\qpkgen$ and $\qenc$.

Next, we present a quantum analogue of classical indistinguishability under chosen-plaintext attacks security (denoted as IND-CPA) for qPKE.
\begin{definition}
    A qPKE scheme is IND-CPA secure if for every $\qpt$ adversary, there exists a negligible function $\epsilon$ such that the probability of winning the IND-CPA security game (see Game \ref{game:cpa}) is at most $\frac{1}{2}+\epsilon(\secpar)$.
    \label{def:EQPK_cpa_security}
\end{definition}

\begin{game}
 \caption{IND-CPA security game for encryption with quantum public key schemes.} \label{game:cpa}
\begin{algorithmic}[1]
    \State  The challenger generates $\dk \gets \gen(\secparam)$.
    \State The adversary gets $\secparam$ as an input, and oracle access to $\qpkgen(\dk)$, and sends $m_0,m_1$ of the same length to the challenger.
    \State The challenger samples $b\in_R \{0,1\}$, generates $\ket{\qpk}\gets \qpkgen(\dk)$ and sends $c\gets \qenc(\ket{\qpk},m_b)$ to the adversary.
    \State The adversary outputs a bit $b'$.
\end{algorithmic}
We say that the adversary wins the game (or alternatively, that the outcome of the game is 1) iff $b=b'$.
 \end{game}
Note that this is the standard CPA-security game of a public-key encryption scheme, with the exception that the adversary can receive polynomially many copies of $\ket{\qpk}$, by making several calls to the $\qpkgen(\dk)$ oracle.

In  the classical setting, there is no need to provide access to an encryption oracle since the adversary can use the public key to apply the encryption herself. In the quantum setting, this is not the case: as we will see, the quantum public key might be measured, and the ciphertexts might depend on the measurement outcome.
This motivates a stronger definition in which the adversary gets oracle access to the encryption, denoted as IND-CPA-EO security.

\begin{definition}
\label{def:single-challenge-cpa-eo-security}
A qPKE scheme is (single-challenge) IND-CPA-EO secure if for every $\qpt$ adversary, there exists a negligible function $\epsilon$ such that the probability of winning the IND-CPA-EO security game (see Game \ref{game:single-challenge-cpa-eo}) is at most $\frac{1}{2}+\epsilon(\secpar)$.
\end{definition}
\begin{game}
 \caption{(Single-challenge) Chosen plaintext attack with an encryption oracle (IND-CPA-EO) security game for encryption with quantum public key schemes.} \label{game:single-challenge-cpa-eo}
\begin{algorithmic}[1]
  \State The challenger generates $\dk \gets \gen(\secparam)$.
  \State The adversary gets $\secparam$ as an input, and oracle access to
  $\qpkgen(\dk)$.
  \State The challenger generates $\ket{\qpk}\gets \qpkgen(\dk)$.
  Let
  \(\qpk_{1} \coloneqq \ket{\qpk}\).
  \State For $i=1,\ldots,\ell$, the adversary creates a classical message $m_i$ and send it to the challenger.
  \State The challenger computes $(\qc_i,\qpk_{i+1}) \gets \qenc(\qpk_i,m_i)$ and send $\qc_i$ to the adversary.
  \State The adversary sends two messages $m'_0,m'_1$ of the same length to the challenger.
  \State The challenger samples $b\in_R \{0,1\}$, computes
  $(\qc^{*}, \qpk_{l+2}) \gets \qenc(\qpk_{\ell+1},m'_b)$ and sends \(\qc^{*}\) to the adversary.
  \State For $i=\ell+2,\ldots,\ell'$, the adversary creates a classical message $m_i$ and send it to the challenger.
  \State The challenger computes $(\qc_i,\qpk_{i+1}) \gets \qenc(\qpk_i,m_i)$ and send $\qc_i$ to the adversary.
  \State The adversary outputs a bit $b'$.
\end{algorithmic}
We say that the adversary wins the game (or alternatively, that the outcome of
the game is 1) iff $b=b'$.
\end{game}

\begin{definition}
\label{def:many-challenge-cpa-eo-security}
A qPKE scheme is (multi-challenge) IND-CPA-EO secure if for every $\qpt$ adversary, there exists a negligible function $\epsilon$ such that the probability of winning the IND-CPA-EO security game (see Game \ref{game:many-challenge-cpa-eo}) is at most $\frac{1}{2}+\epsilon(\secpar)$.
\end{definition}
\begin{game}
 \caption{(Multi-challenge) Chosen plaintext attack with an encryption oracle (IND-CPA-EO) security game for encryption with quantum public key schemes.}
 \label{game:many-challenge-cpa-eo}
\begin{algorithmic}[1]
  \State The challenger generates $\dk \gets \gen(\secparam)$.
  \State The adversary gets $\secparam$ as an input, and oracle access to
  $\qpkgen(\dk)$.
  \State The challenger generates $\ket{\qpk}\gets \qpkgen(\dk)$.
  Let
  \(\qpk_{1} \coloneqq \ket{\qpk}\).
  \State For $i=1,\ldots,\ell$, the adversary creates a classical message $m_i$ and send it to the challenger.
  \State The challenger computes $(\qc_i,\qpk_{i+1}) \gets \qenc(\qpk_i,m_i)$ and send $\qc_i$ to the adversary.
  \State The adversary sends two messages $m'_0,m'_1$ of the same length to the challenger.
  \State The challenger samples $b\in_R \{0,1\}$, computes
  $(\qc^{*}, \qpk_{l+2}) \gets \qenc(\qpk_{\ell+1},m'_b)$ and sends \(\qc^{*}\) to the adversary.
  \State For $i=\ell+2,\ldots,\ell'$, the adversary creates a classical message $m_i$ and send it to the challenger.
  \State The challenger computes $(\qc_i,\qpk_{i+1}) \gets \qenc(\qpk_i,m_i)$ and
  send $\qc_i$ to the adversary.
  \State The challenger and the adversary can repeat step 3 - 9 polynomially
  many times.
  For the \(i\)-th repetition, let \(\qpk_{1, i}\) be
  \(\qpk_{l'+1, i-1}\).
  \State The challenger and the adversary can repeat step 3 - 10 polynomially
  many times.
  \State The adversary outputs a bit $b'$.
\end{algorithmic}
We say that the adversary wins the game (or alternatively, that the outcome of
the game is 1) iff $b=b'$.
\end{game}

\begin{theorem}
Single-challenge security (\cref{def:single-challenge-cpa-eo-security}) implies multi-challenge security (\cref{def:many-challenge-cpa-eo-security}).
\end{theorem}
\begin{proof}
  Recall that classically, single-challenge IND-CPA implies multi-challenge IND-CPA, see, e.g.,~\cite[Theorem 3.24]{KL14}. Following the exact same argument works in our case as well: single-challenge IND-CPA-EO implies multi-challenge IND-CPA-EO.
\end{proof}

\section{Constructions}
\subsection{QPKE with Classical ciphertext from OWFs}
\label{sect:qpke-from-owf}
We show a construction based on the existence of post-quantum one-way functions. Recalls that post-quantum one-way functions imply a qPRF~\cite{Zha12}, and IND-CPA\footnote{In fact, this can be strengthened to IND-qCCA security, but this is unnecessary for our application.} symmetric encryption~\cite{BZ13}.
\begin{figure}[phtp]
%\framebox{\parbox{\hsize}{
\noindent\textbf{Assumes:} A PRF family $\{f_k\}_{k}$, and a symmetric encryption scheme $\{\enc,\dec\}$.\\
\noindent$\gen(\secparam)$
\begin{compactenum}
\item $\dk \gets_R \{0,1\}^\secpar$.
\end{compactenum}

\noindent $\qpkgen(\dk) $
\begin{compactenum}
\item Output $ \ket{\qpk}=\frac{1}{\sqrt{2^\secpar}} \sum_{x \in \{0,1\}^\secpar}\ket{x}\ket{f_\dk(x)}$.
\end{compactenum}

\noindent $\qenc(\ket{\qpk},m)$
\begin{compactenum}
\item Measure both registers of $\ket{\qpk}$ in the standard basis. Denote the result as $x$ and $y$.
\item Output $c=(x,\enc(y, m))$.
\end{compactenum}

\noindent $\dec_\dk(c)$
\begin{compactenum}
\item Interpret $c$ as $(x,z)$
\item Set $y:=f_\dk(x)$.
\item Output $m=\dec(y, z)$.
\end{compactenum}

\caption{An encryption scheme with quantum public keys.}
\label{fig:public_key_encryption_from_prf}
\end{figure}

\begin{theorem}\label{lem:security-prf-scheme}
 Assuming the existence of quantum-secure PRF family $\{f_k\}_k$ and post-quantum IND-CPA symmetric-key encryption scheme $(\enc,\dec)$, any QPT adversary $\adv$ wins the IND-CPA-EO game for the scheme presented in Figure~\ref{fig:public_key_encryption_from_prf} with advantage at most $\negl$.
\end{theorem}
\begin{proof}
In order to prove our results, we define the following Hybrids.

\noindent\textbf{Hybrid $H_0$.} The original security game as defined in~\cref{game:single-challenge-cpa-eo}.

\noindent\textbf{Hybrid $H_1$.} Same as Hybrid $H_0$, except that the challenger, instead of measuring $\ket{\qpk}$ when the adversary queries the encryption oracle for the first time, the challenger measures this state before providing the copies of $\ket{\qpk}$ to the adversary.
Note that by measuring $\ket{\qpk}$ in the computational basis, the challenger
would obtain a uniformly random string $x^*$ (and the corresponding
$f_\dk(x^*)$).

\noindent\textbf{Hybrid $H_2$.} Same as Hybrid $H_1$, except that the challenger
samples $x^*$ as in the previous hybrid, and instead of providing $\ket{\qpk}$
to the adversary, she provides
\[\ket{\qpk'} = \frac{1}{\sqrt{2^\lambda -1}}\sum_{x \in \{0,1\}^*, x \ne x^*}  \ket{x}\ket{f_{\dk}(x)}.\]
Moreover, the challenger provides ciphertexts $(x^*,\enc(f_{\dk}(x^*), m))$ for the chosen messages $m$.
We note that this state $\ket{\qpk'}$ can be efficiently prepared by computing
the functions $\delta_{x, x^*}$ over the state $\sum_x \ket{x}$ in superposition and
measuring the output register.
With overwhelming probability, the post-measurement state is
$\sum_{x \neq x^*}\ket{x}$.

\noindent\textbf{Hybrid $H_3$.}
Same as Hybrid $H_2$, except the challenger uses a random function \(H\) in
place of \(f_{\dk}\), and provides
$\ket{\qpk'} = \frac{1}{\sqrt{2^\lambda -1}}\sum_{x \in \{0,1\}^*, x \ne x^*} \ket{x}\ket{H(x)}$.
Moreover, for each encryption query, the challenger uses \(H(x^{*})\) as
\(\qpk_{i}\), and answers the query with $(x^*, \enc(H(x^{*}), m))$
for the chosen message $m$.

\noindent\textbf{Hybrid $H_4$.}
Same as Hybrid $H_3$, except the challenger samples uniformly at random a string
\(z\), and uses \(z\) as \(\qpk_{i}\).
The answer to each encryption query is now $(x^*, \enc(z, m))$ for the chosen
message $m$.

We will now show that, given our assumptions, Hybrids $H_i$ and $H_{i+1}$ are indistinguishable except with probability at most $\negl$ and that every polynomial-time adversary wins $H_4$ with advantage at most $\negl$. We can then use triangle inequality to prove the security of \Cref{lem:security-prf-scheme}.

\begin{lemma}\label{lem:hybrid1_prf}
No adversary can distinguish Hybrid $H_0$ and  Hybrid $H_1$ with non-zero advantage.
\end{lemma}
\begin{proof}
    Notice that the operations corresponding to the challenger's measurement of $\ket{\qpk}$ and the creation of the copies of $\ket{\qpk}$ given to the adversary commute. In this case, we can swap the order of these operations and the outcome is exactly the same.
\end{proof}

\begin{lemma}\label{lem:hybrid2_prf}
No adversary can distinguish Hybrid $H_1$ and Hybrid $H_2$ with non-negligible advantage.
\end{lemma}
\begin{proof}
  Notice that distinguishing the two adversaries imply that we can distinguish
  the following quantum states $\ket{\qpk}^{\otimes p} \otimes\ket{x^*}$ and
  $\ket{\qpk'}^{\otimes p} \otimes\ket{x^*}$, but these two quantum states have
  $1 - \negl$ trace-distance for any polynomial $p$.
  Therefore, this task can be performed with success at most $\negl$.
\end{proof}

\begin{lemma}\label{lem:hybrid3_prf}
No $\qpt$ adversary can distinguish Hybrid $H_2$ and  Hybrid $H_3$ with non-negligible advantage.
\end{lemma}
\begin{proof}
Suppose that there exists an adversary $\adv$ such that
$\Pr[F_1] - \Pr[F_2] \geq \frac{1}{p(\lambda)}$ for some polynomial $p$, where
$F_1$ is the event where $\adv$ outputs $1$ on $H_2$ and
$F_2$ is the event where $\adv$ outputs $1$ on $H_3$. Then, we show that we can construct an adversary $\adv'$ that can distinguish the PRF family from random.

$\adv'$ will behave as the challenger in the CPA-EO game and instead of computing $f_{\dk}$ to create $\ket{\qpk'}$ and answer the encryption queries, she queries the oracle $O$ (that is either a PRF or a random function). $\adv'$ then outputs $1$ iff $\adv$ outputs $1$.

Notice that if $O$ is a PRF, then the experiment is the same as Hybrid $H_2$. On the other hand, if $O$ is a random oracle, the experiment is the same as Hybrid $H_3$.

In this case, we have that
\begin{align}
&\Pr_{O \sim \mathcal{F}}[\adv'^{O}() = 1] - \Pr_{\dk}[\adv'^{f_{\dk}}() = 1] \\
&=
\Pr[F_1]
- \Pr[F_2] \\
&\geq \frac{1}{p(\lambda)}.
\end{align}
\end{proof}

\begin{lemma}\label{lem:hybrid4_prf}
  No adversary can distinguish Hybrid $H_3$ and Hybrid $H_4$ with non-zero
  advantage.
\end{lemma}
\begin{proof}
  Since the adversary never gets the evaluation of \(H(x^{*})\) (as \(x^{*}\)
  was punctured from all \(\ket{\qpk}\)), the distributions of the two hybrid are
  identical.
\end{proof}

\begin{lemma}\label{lem:security_hybrid3_prf}
 Any polynomially-bounded adversary wins the game in Hybrid $H_4$ with an advantage at most $\negl$.
\end{lemma}
\begin{proof}
  Suppose that there exists an adversary $\adv$ such that wins the game in
  Hybrid $H_4$ with advantage $\frac{1}{p(\lambda)}$ for some polynomial $p$.
  Then, we show that we can construct an adversary $\adv'$ that can break IND-CPA
  security of the symmetric-key encryption scheme with the same probability.

  $\adv'$ will simulate $\adv$ and for that, she picks $x^*$ and $z$, creates
  $\ket{\qpk'} = \frac{1}{\sqrt{2^\lambda -1}}\sum_{x \in \{0,1\}^*, x \ne x^*} \ket{x}\ket{H(x)}$
  using the compressed oracle technique~\cite{Zhandry19} and uses oracle provided by the IND-CPA game of the
  symmetric-key encryption scheme for answering the encryption oracles.
  $\adv'$ will output $1$ iff $\adv$ outputs $1$.
  We note that the encryption key \(z\) is sampled uniformly at random
  independently of all other variables.
  We have that the winning probability of $\adv'$ in the IND-CPA game is the
  same of $\adv$ in the IND-CPA-EO game.
\end{proof}
\end{proof}

\subsection{QPKE with Classical Ciphertexts from PRFSPDs}
\label{sect:qpke-from-prfspd}
In this section, we propose a construction for qPKE from pseudo-random function-like states with proof of destruction.

\begin{figure}[phtp]
%\framebox{\parbox{\hsize}{
\noindent\textbf{Assumes:} A PRFSPD family $\{\ket{\psi_{\dk,x}}\}_{\dk,x}$ and a quantum symmetric encryption scheme with classical ciphers $\{\enc,\dec\}$.\\
\noindent$\gen(\secparam)$
\begin{compactenum}
\item $\dk \gets_R \{0,1\}^\secpar$.
\end{compactenum}

\noindent $\qpkgen(\dk) $
\begin{compactenum}
\item Output $ \ket{\qpk}=\bigotimes_{i \in [\lambda]} \frac{1}{\sqrt{2^\secpar}} \sum_{x^{(i)} \in \{0,1\}^\secpar}\ket{x^{(i)}}\ket{\psi_{\dk,x^{(i)}}}$.
\end{compactenum}

\noindent $\qenc(\ket{\qpk},m)$ for $m \in\{0,1\}$
\begin{compactenum}
\item Let $\ket{\qpk_i} \coloneqq \frac{1}{\sqrt{2^\secpar}} \sum_{x^{(i)} \in \{0,1\}^\secpar}\ket{x^{(i)}}\ket{\psi_{\dk,x^{(i)}}}$, and write $\ket{\qpk}$ as $\ket{\qpk}=\bigotimes_{i \in [\lambda]} \ket{\qpk_i}$.
\item Measure the left registers of $\ket{\qpk_i}$ to obtain classical strings $x^{(i)}$. Denote the post-measurement states as $\ket{\psi'_i}$.
\item Set $y^{(i)}\gets \del(\ket{\psi'})$.
\item Pick $k = \{0,1\}^\lambda$ and $r^{(i)} = \{0,1\}^{\size{y^{(i)}}}$ uniformly at random.
\item Set $\tilde{y}^{(i)} = \begin{cases}r^{(i)}&, \text{if } k_i = 0 \\ y^{(i)}&, \text{if } k_i = 1 \end{cases}$.
\item Output $\left(\enc(k, m), \left( (x^{(i)},\tilde{y}^{(i)})\right)_i\right)$
\end{compactenum}

\noindent $\dec_\dk(c)$
\begin{compactenum}
\item Interpret  $c$ as $\left(c', \left( (x^{(i)},\tilde{y}^{(i)})\right)_i\right)$
\item Let $k'_i = \ver(x^{(i)},\tilde{y}^{(i)})$.
\item Output $\dec(k', c')$
\end{compactenum}

\caption{An encryption scheme with quantum public keys.}
\label{fig:public_key_encryption_from_prfspd}
\end{figure}

In this section, we construct a qPKE scheme based on PRFSPD. For that, we need the following result that builds \emph{symmetric}-key encryption from such an assumption.

\begin{proposition}[\cite{BBSS23}] If  quantum-secure PRFSPD exists, then there exists a quantum CPA symmetric encryption with classical ciphertexts.    
\end{proposition}

\begin{theorem}\label{lem:security-prfspd-scheme}
If quantum-secure  PRFSPD with super-logarithmic input size exists, then there exists a public-key encryption  with classical ciphertexts which is IND-CPA-EO secure.
\end{theorem}

\begin{proof}
Our construction is given in \cref{fig:public_key_encryption_from_prfspd}. It uses a PRFSPD, as well as  a quantum CPA \emph{symmetric} encryption with classical 
ciphertexts. Such symmetric encryption is known to exist, based on PRFSPD:

In order to prove our results, we define the following Hybrids.

\noindent\textbf{Hybrid $H_0$.} This is the original security game.

\noindent\textbf{Hybrid $H_1$.} Same as Hybrid $H_0$, except that the challenger picks ${x^{(i)}}^*$ uniformly at random, instead of providing $\ket{\qpk}$ to the adversary, she provides
\[\ket{\qpk'}=\bigotimes_{i \in [\lambda]} \frac{1}{\sqrt{2^\secpar - 1}} \sum_{x^{(i)} \in \{0,1\}^\secpar, x^{(i)}\ne {x^{(i)}}^*}\ket{x^{(i)}}\ket{\psi_{\dk,x^{(i)}}}.\]
The challenger uses the states $\ket{{x^{(i)}}^*}\ket{\psi_{\dk,{x^{(i)}}^*}}$ to encrypt the challenge.

\noindent\textbf{Hybrid $H_2$.} Same as Hybrid $H_1$, but to answer the encryption queries, the challenger picks each $\tilde{y}^{(i)}$ uniformly at random and answers the encryption queries with $\left(\enc(k, m), \left( (x^{(i)},\tilde{y}^{(i)})\right)_i\right)$

We will now show that, given our assumptions, Hybrids $H_i$ and $H_{i+1}$ are indistinguishable except with probability at most $\negl$ and that every polynomial-time adversary wins $H_3$ with advantage at most $\negl$. We can then use triangle inequality to prove the security of \Cref{lem:security-prfspd-scheme}.

\begin{lemma}
No adversary can distinguish Hybrid $H_0$ and  Hybrid $H_1$ with non-zero advantage.
\end{lemma}
\begin{proof}
The proof follows analogously to \Cref{lem:hybrid1_prf,lem:hybrid2_prf} and the triangle inequality.
\end{proof}

\begin{lemma}
No adversary can distinguish Hybrid $H_2$ and  Hybrid $H_3$ with non-negligible advantage.
\end{lemma}
\begin{proof}
Suppose that there exists an adversary $\adv$ such that
$\Pr[F_1] - \Pr[F_2] \geq \frac{1}{p(\lambda)}$ for some polynomial $p$, where
$F_1$ is the event where $\adv$ outputs $1$ on $H_2$ and
$F_2$ is the event where $\adv$ outputs $1$ on $H_3$. Then, we show that we can construct an adversary $\adv'$ that can break the PRFSPD family scheme.

$\adv'$ will behave as the challenger in the CPA-EO game and instead of computing $\ket{\qpk'}$ by herself, she queries the oracle $O$ (that on input $\ket{x}\ket{0}$ answers either a PRFSPD $\ket{x}\ket{\psi_{\dk,x}}$ or a Haar random state $\ket{x}\ket{\vartheta_x}$). Then, $\adv'$ picks a random bit $b$ and performs as follows:
\begin{itemize}
    \item if $b = 0$, $\adv'$ answers the encryption queries as Hybrid $H_2$
    \item if $b = 1$, $\adv'$ picks each $\tilde{y}^{(i)}$ uniformly at random and answers the encryption queries with $\left(Enc_k(m), \left( (x^{(i)},\tilde{y}^{(i)})\right)_i\right)$
\end{itemize}

$\adv'$ then outputs $1$ iff $\adv$ outputs $1$.

Notice that if $b=0$ and $O$ answers a PRFSPD, then the experiment is the same as Hybrid $H_2$. On the other hand, if $b = 0$ and $O$ is a random oracle or if $b=1$, the experiment is the same as Hybrid $H_3$.
Let us define the event $E_1$ be the event where $\adv'$ outputs $1$ if $O$ is a PRFSPD, and $E_2$ be the event where $\adv'$ outputs $1$ if $O$ answers with Haar random states.

In this case, we have that $\adv'$ distinguishes the PRFS family from Haar random states with probability
\begin{align}
&\Pr[E_1] - \Pr[E_2] \\
&=
\frac{1}{2}\left(
\Pr[E_1 \wedge b=0 ] +
\Pr[E_1 \wedge b=1 ]
- \Pr[E_2 \wedge b=0 ]
- \Pr[E_2 \wedge b=1] \right)\\
&=
\frac{1}{2}\left(
\Pr[F_1]
+ \Pr[F_2]
-  \Pr[F_2]
- \Pr[F_2]
\right) \\
&\geq \frac{1}{2p(\lambda)}.
\end{align}
\end{proof}

\begin{lemma}
 Any polynomially-bounded adversary wins the game in Hybrid $H_3$ with an advantage at most $\negl$.
\end{lemma}
\begin{proof}
Suppose that there exists an adversary $\adv$ such that wins the game in Hybrid $H_3$ with advantage $\frac{1}{p(\lambda)}$ for some polynomial $p$. Then, we show that we can construct an adversary $\adv'$ that can break IND-CPA security of the symmetric-key encryption scheme.

$\adv'$ will simulate $\adv$ and for that, she picks $\dk$ and creates $\ket{\qpk'}$ and in order to answer the encryption queries, they use the encryption oracle provided by the IND-CPA game. $\adv'$ will output $1$ iff $\adv$ outputs $1$. We have that the winning probability of $\adv'$ in its IND-CPA game is the same of $\adv$ in its IND-CPA-EO game.
\end{proof}
\end{proof}

\subsection{QPKE with Quantum Ciphertexts from PRFSs}
\label{sect:qpke-from-prfs}

We finally present our third scheme for qPKE, whose security is based on the existence of PRFS with super-logarithmic input size.

\begin{theorem}
    The construction in \cref{fig:public_key_encryption_from_prfs} is IND-CPA secure (see \cref{def:EQPK_cpa_security}), assuming $\{\ket{\psi_{k,x}}\}_{k,x}$ is a PRFS with super-logarithmic input-size.
\end{theorem}
The proof of this theorem uses the same proof strategy of~\cref{lem:security-prf-scheme} and~\cref{lem:security-prfspd-scheme}, the only difference is that here the scheme is only IND-CPA secure, while the previous ones are IND-CPA-EO secure.

\begin{figure}[H]
%\framebox{\parbox{\hsize}{
\noindent\textbf{Assumes:} A PRFS family $\{\ket{\psi_{\dk,x}}\}_{\dk,x}$ with super-logarithmic input-size.\\
\noindent$\gen(\secparam)$
\begin{compactenum}
\item $\dk \gets_R \{0,1\}^\secpar$.
\end{compactenum}

\noindent $\qpkgen(\dk) $
\begin{compactenum}
\item Output $ \ket{\qpk}=\frac{1}{\sqrt{2^\secpar}} \sum_{x \in \{0,1\}^\secpar}\ket{x}\ket{\psi_{\dk,x}}$.
\end{compactenum}

\noindent $\qenc(\ket{\qpk},m)$ for $m \in\{0,1\}$
\begin{compactenum}
\item Measure left register, denoted by $x$. Let $\ket{\phi} = \ket{\psi_{\dk,x}}$ if $m=0$, and a maximally mixed state otherwise.
\item Output $c=(x,\ket{\phi})$.
\end{compactenum}

\noindent $\dec_\dk(c)$
\begin{compactenum}
\item Interpret $c$ as $(x,\ket{\phi})$
\item Output 0 if $\ket{\phi}=\ket{\psi_{\dk,x}}$, otherwise output 1.
\end{compactenum}

\caption{An encryption scheme with quantum public keys based on a PRFS.}
\label{fig:public_key_encryption_from_prfs}
\end{figure}

% \clearpage

%%%%%%%%%%%%%%%%%%%%%%%%%%%%
\ifnum\masterthesis=0
    \subsection*{Acknowledgments}
\fi
%%%%%%%%%%%%%%%%%%%%%%%%%%%%
%NON-ANON PART

\ifnum\anon=0
%ANON PART
We wish to thank Prabhanjan Ananth and Umesh Vazirani for related discussions.
ABG and QHV are supported by ANR JCJC TCS-NISQ ANR-22-CE47-0004, and by the PEPR integrated project EPiQ ANR-22-PETQ-0007 part of Plan France 2030.
OS was supported by the Israeli Science Foundation (ISF) grant No. 682/18 and 2137/19, and by the Cyber Security Research Center at Ben-Gurion University.

\BeforeBeginEnvironment{wrapfigure}{\setlength{\intextsep}{0pt}}
\begin{wrapfigure}{r}{90px}
    %\centering
    \includegraphics[width=40px]{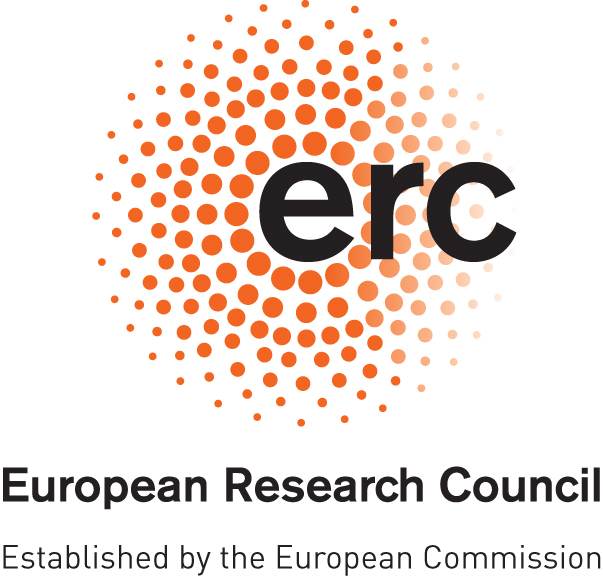}  \includegraphics[width=40px]{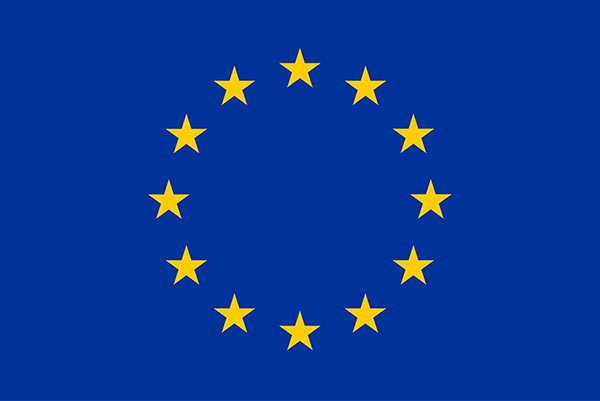}
\end{wrapfigure} OS has received funding from the European Research Council (ERC) under the European Union’s Horizon 2020 research and innovation programme (grant agreement No 756482).

\fi
%%%%%%%%%%%%%% HANDLING OF BIBLIOGRAPHY
%\nocite{*} %to check if the added bib entries compile correctly with the old file.
\ifnum\sigconf=1
    \bibliographystyle{ACM-Reference-Format}
\else
    \ifnum\cryptology=1
        \bibliographystyle{abbrv}
    \else
        \bibliographystyle{alphaabbrurldoieprint}

    \fi
\fi

\ifnum\masterthesis=0
    \ifnum\smallbib=1
        {\footnotesize \bibliography{main} }
    \else
        \bibliography{main}
    \fi
\fi

\appendix
\ifnum\shownomenclature=1
\printnomenclature[1in]
%There is a label which is created automatically in the macros of the following form: \label{sec:nomenclature}
\fi
\ifnum\draft=1
This appendix is shown only if draft mode is enabled.
    \input{nextver}
\fi
\ifnum\masterthesis=1
    \ifnum\smallbib=1
        {\footnotesize \bibliography{main} }
    \else
        \bibliography{main}
    \fi
\fi

\end{document}